\documentclass{llncs}

\usepackage{amsmath}
\usepackage{url}
\usepackage{amsfonts,amssymb}
\usepackage{xspace}
\usepackage[noend]{algorithm}
\usepackage[noend]{algpseudocode}
\usepackage{tikz}
\usepackage{pgfplots}

\newcommand{\algotask}{Algorithm \textsc{TaskScheduling}($b$)\xspace}

\newcommand{\mrs}{Algorithm \textsc{MRS}\xspace}

\newcommand{\mr}{Map-Reduce scheduling\xspace}
\newcommand{\msr}{Map-Shuffle-Reduce scheduling\xspace}
\begin{document}

\title{Scheduling MapReduce Jobs and Data Shuffle on Unrelated Processors\thanks{The authors were partially supported by the European Social Fund and Greek national resources under Thales-DELUGE and Heracleitus II programs. A short extended abstract of this work, including partial results, appeared in EDBT/ICDT 2014 Workshop on Algorithms for MapReduce and Beyond.}}

\author{
Dimitris Fotakis\inst{1}
\and Ioannis Milis\inst{2}
\and Emmanouil Zampetakis\inst{1}
\and Georgios Zois\inst{2,3}
}

\institute{School of Electrical and Computer Engineering, National Technical University of Athens, Greece\\
\email{fotakis@cs.ntua.gr,}
\email{mzampet@corelab.ntua.gr}
\and Department of Informatics, Athens University of Economics and Business, Greece\\
\email{\{milis, georzois\}@aueb.gr}
\and Sorbonne Universit\'{e}s, UPMC Univ Paris 06, UMR 7606, LIP6, F-75005, France\\
\email{\{Georgios.Zois\}@lip6.fr}
}

\maketitle

\begin{abstract}

We propose constant approximation algorithms for generalizations of the Flexible Flow Shop (FFS) problem  which form a realistic model for non-preemptive scheduling in MapReduce systems.  Our results concern the minimization of the total weighted completion time of a set of MapReduce jobs on unrelated processors and improve substantially on the model proposed by Moseley et al. (SPAA 2011) in two directions. First, we consider each job consisting of multiple Map and Reduce tasks, as this is the key idea behind MapReduce computations, and we propose a constant approximation algorithm.
Then,  we introduce into our model the crucial cost of data shuffle phase, i.e., the cost for the transmission of intermediate  data from Map to Reduce tasks. In fact, we model this phase by   an additional set of Shuffle tasks for each job and we manage to keep the same approximation ratio  when they are scheduled on the same processors with the corresponding Reduce tasks and to provide also a constant ratio when they are scheduled on different processors.  This is the most general setting of the FFS problem (with a special third stage) for which a constant approximation ratio is known.
\end{abstract}

\section{Introduction}

The widespread use of MapReduce ~\cite{DeanG04} to implement massive parallelism for data intensive computing motivates the study of  new challenging shop scheduling problems. Indeed, a MapReduce job consists of a set of Map tasks and a set of Reduce tasks that can be executed simultaneously, provided that no Reduce task of a job can start execution before all the Map tasks of this job are completed. Moreover, a significant part of the processing cost in MapReduce applications is the communication cost due to the transmission of intermediate data from Map tasks to Reduce tasks (referred as data shuffle, see e.g., \cite{AU09,ASSU13,Ullman12}).
To exploit the inherent parallelism, the scheduler of such a system, which operates in centralized manner, has to efficiently assign and schedule Map and Reduce tasks to the available processors. In this context, standard shop scheduling problems are revisited to capture key constraints and singularities of MapReduce systems. In fact, a few results have been recently proposed based on simplified abstractions and resulting in known variants of the classical Open Shop and Flow Shop scheduling problems  ~\cite{ChangKKLLM11,ChenKL12,MoseleyDKS11}.

In this paper, we significantly generalize the Flexible Flow Shop (FFS) model for MapReduce scheduling proposed in~\cite{MoseleyDKS11}. Recall that in the FFS problem, we are given a set of jobs, each consisting of a number of tasks (each task corresponds to a stage), to be scheduled on a set of parallel processors dedicated to each stage. The jobs should be executed in the same fixed order of stages, without overlaps between tasks (stages) of the same job. Our generalization extends substantially the model proposed in~\cite{MoseleyDKS11} by taking into account all the important constraints of MapReduce systems:
(a) each job has multiple tasks in each stage; (b) the assignment of tasks to processors is flexible; (c) there are dependencies between Map and Reduce tasks; (d) the processors are unrelated to capture data locality; and (e) there is a significant communication cost for the data shuffle. Our goal is to find a non-preemptive schedule  minimizing the standard objective of total weighted completion time for a set of MapReduce jobs.

\noindent{\bf Motivation.}
MapReduce has been established as the standard programming model to implement massive parallelism in large data centers~\cite{DeanG04}. Applications of MapReduce such as search indexing, web analytics, and data mining involve the concurrent execution of several MapReduce jobs on a system like Google's MapReduce~\cite{DeanG04} or Apache Hadoop ~\cite{Hadoop}.
When a MapReduce job is executed, a number of Map and Reduce tasks are created. Each Map task operates on a portion of the input elements, translating them into a number of key-value pairs and, after an intermediate process, all pairs having the same key are available to a Reduce task, which operates on the values associated with that key and generates the final result. The basic idea behind MapReduce computation is that each job is split into a large number of Map and Reduce tasks that can be executed in parallel (see e.g., \cite{AU09,KSV10,ASSU13}).
In addition, a significant cost when running a MapReduce job is that of data shuffle, i.e., the transmission of intermediate data of a job from Map tasks to  Reduce tasks. This cost affects crucially the performance of MapReduce systems (e.g., bandwidth bottleneck~\cite{ChenKL12}, high wall-clock time~\cite{Ullman12}) and usually dominates the computation cost of Map and  Reduce tasks (see e.g.~\cite{AU09,ASSU13}).  In terms of scheduling,
this makes the problem more intricate and important for system performance.


\noindent{\bf Related Work.} Known results for the FFS problem concern the two-stage case
on parallel identical processors. For the makepsan objective a PTAS is known ~\cite{SchuurmanW00}, while for the total weighted completion time objective, a simple 2-approximation algorithm was proposed in~\cite{GonzalezS78}, for the special case where   each stage has to be executed on a single processor.  For the latter case, recently in~\cite{MoseleyDKS11}
the authors proposed a QPTAS which becomes a PTAS for a fixed number of processing times of tasks.

In the MapReduce context, most of the previous work concerns the experimental evaluation of scheduling heuristics, from the viewpoint of finding good trade-offs between different practical criteria (see e.g., \cite{YooS11} and the references therein). From a theoretical point of view all known results ~\cite{ChangKKLLM11,ChenKL12,MoseleyDKS11} concern the minimization of  total weighted completion time. Chang et al. ~\cite{ChangKKLLM11}
 proposed approximation results using a very simple model  equivalent to the so called \emph{concurrent open shop} problem ~\cite{MastrolilliQSSU10}, where the dependencies between Map and Reduce tasks are missing and the assignment of tasks to processors is given.
Chen et al.~\cite{ChenKL12} proposed a more general model taking into account task dependencies but also  assuming  that tasks are  preassigned to  processors.  For this
 restricted model, they  presented  an LP-based $8$-approximation algorithm. Moreover, they managed to incorporate the data shuffle in their model and to derive a   $58$-approximation algorithm. Finally,
Moseley et al.~\cite{MoseleyDKS11} introduced the relation with the FFS problem and studied the cases of both identical and unrelated processors.   For identical processors, they presented  a $12$-approximation algorithm, and a $O(1/\epsilon^2)$-competitive online algorithm, for any $\epsilon \in (0, 1)$, under $(1+\epsilon)$-speed augmentation.
For unrelated processors  they studied  the very restricted  case where each job has a single Map and a single Reduce task, and presented a $6$-approximation algorithm and a $O(1/\epsilon^5)$-competitive online algorithm, for any $\epsilon \in (0, 1)$, under $(1+\epsilon)$-speed augmentation.

\noindent{\bf Our Results.} 
We present constant approximation algorithms  which substantially generalize the results  of \cite{MoseleyDKS11} for  MapReduce scheduling on unrelated processors towards two directions motivated by real MapReduce systems. In fact, we deal with jobs consisting of multiple Map and Reduce tasks and also incorporate the shuffle phase into our setting. As it has been observed in \cite{MoseleyDKS11}, new ideas and techniques are required for both these directions.

In Section 2, we present a $54$-approximation algorithm for the \mr problem when jobs consist of multiple Map and Reduce tasks. We first formulate an interval-indexed LP-relaxation for the problem of minimizing the total weighted completion times separately for  Map and  Reduce tasks on unrelated processors. Our LP formulation is inspired by the one proposed by Hall et al.~\cite{HallSSW97} for scheduling a set of single task jobs on unrelated processors under the same objective. However, in our problem, not all the tasks of each job contribute to the objective value, but only the one that finishes last and this makes the analysis of such an LP more difficult.
Recently, Correa et al.~\cite{CorreaSV12} proposed and analyzed a similar LP-relaxation for a more general problem, where, instead of jobs consisting of tasks, they are given a set of orders of jobs and the completion time of each order is specified by the completion of the job that finishes last.
Since scheduling multitask Map and Reduce jobs separately is quite similar to the setting considered in \cite{CorreaSV12}, we can use their approximation result for scheduling separately the Map and Reduce tasks. Next, we concatenate the two schedules into a single one respecting the task dependencies, by extending the ideas in~\cite{MoseleyDKS11} for single task jobs.


%
%

In Section~3, we incorporate the data shuffle phase into our model by  introducing  an additional set of \emph{Shuffle tasks} each one associated with a   communication cost (time).  When the Shuffle tasks are scheduled on the same processors as the corresponding Reduce tasks we are able to keep the same $54$-approximation ratio for the Map-Shuffle-Reduce scheduling problem.  Moreover, we also prove an $81$-approximation ratio when the  Shuffle tasks are allowed to  be executed on different processors than their corresponding Reduce tasks. To the best of our knowledge, this is the most general setting of the FFS problem (with a special third stage) for which a constant  approximation guarantee is known.

\noindent{\bf Problem statement and notation.}
In the sequel we consider a set $\mathcal{J}=\{1,2,\ldots,n\}$ of $n$ MapReduce jobs to be executed
on a set $\mathcal{P}=\{1,2,\ldots,m\}$ of $m$ unrelated processors.
Each job is available at time zero, is associated with a positive weight $w_j$
and consists of a set of Map tasks and a set of Reduce tasks. Let $\mathcal{M}$  and $\mathcal{R}$ be the set of all Map and all Reduce tasks respectively.
Each task is denoted by $\mathcal{T}_{k,j}\in \mathcal{M}\cup\mathcal{R}$, where $k \in N$ is the task index of job $j\in \mathcal{J}$
and is associated with a vector of non-negative processing times $\{p_{i,k,j}\}$, one for each processor $i\in \mathcal{P}_b$, where $b\in\{\mathcal{M},\mathcal{R}\}$. Let $\mathcal{P}_{\mathcal{M}}$ and $\mathcal{P}_{\mathcal{R}}$ be the set
of Map and the set of Reduce processors  respectively.
For convenience, we assume that $\mathcal{P}_{\mathcal{M}}\cap\mathcal{P}_{\mathcal{R}} = \emptyset$, however we are able to extend our results to the case where the two sets of processors are not necessarily disjoint (or even are identical).
Each job has at least one Map and one Reduce task and
every Reduce task can start its execution after the completion of all Map tasks of the same job.

For a given schedule we denote by $C_j$ and $C_{k,j}$ the completion times of each
job $j \in \mathcal{J}$ and each  task $\mathcal{T}_{k,j}\in \mathcal{M}\cup\mathcal{R}$ respectively.
Note that, due to the precedence constraints between Map and Reduce tasks, $C_j = \max_{\mathcal{T}_{k,j}\in\mathcal{R}} \{C_{k,j}\}$.
By $C_{max}=\max_{j \in \mathcal{J}} \{C_j\}$ we denote the makespan of the schedule, i.e., the completion time of the job which finishes last.
Our goal is to schedule \emph{non-preemptively} all Map tasks on processors of $\mathcal{P}_{\mathcal{M}}$ and all Reduce tasks on processors of $\mathcal{P}_{\mathcal{R}}$, with respect to their precedence constraints,
so as to minimize the total weighted completion time of the schedule, i.e., $\sum_{j\in\mathcal{J}}w_j C_j$. We refer to this problem as \mr problem.

Concerning the complexity of \mr problem,
it generalizes the FFS problem which is is known to be strongly $\mathcal{NP}$-hard~\cite{GareyJS76b}, even when there is a single Map and a single Reduce task that has to be assigned only to one Map and one Reduce processor respectively.

\section{\mr problem }

In this section, we present a $54$-approximation algorithm for the \mr problem. Our algorithm is
executed in the following two steps: (i) it computes a $27/2$-approximate schedule for assigning and scheduling all Map tasks (resp. Reduce tasks) on processors of the set $\mathcal{P}_{\mathcal{M}}$ (resp. $\mathcal{P}_{\mathcal{R}}$) and (ii) it merges the two schedules in one, with respect to the precedence constraints between Map and Reduce tasks of each job. Step (ii) is performed by increasing the approximation ratio by a factor of $4$.

\subsection{Scheduling Map tasks and Reduce tasks}

To schedule separately the Map and Reduce tasks on the processors   $\mathcal{P}_{\mathcal{M}}$ and  $\mathcal{P}_{\mathcal{R}}$, respectively,
we start by formulating  an interval-indexed LP-relaxation for the minimization of the total weighted completion time. Our LP-relaxation is an adaptation to our problem of the standard LP-relaxation proposed by Hall et al.~\cite{HallSSW97} for the problem of minimizing the total weighted completion time on unrelated processors.

For notational convenience, we use an argument $b\in\{\mathcal{M}, \mathcal{R}\}$ to refer either to Map or to Reduce sets of tasks.
We define $(0,t_{\max} = \sum_{\mathcal{T}_{k,j}\in b}\max_{i\in \mathcal{P}_b}p_{i,k,j}]$ to be the time horizon of potential completion times, where $t_{\max}$ is an upper bound on the makespan of a feasible schedule. We discretize the time horizon into intervals $[1,1],(1,(1+\delta)], ((1+\epsilon),(1+\delta)^2],\ldots,((1+\delta)^{L-1},(1+\delta)^L]$, where $\delta\in (0,1)$ is a small constant, and $L$ is the smallest integer such that $(1+\delta)^{L-1} \geq t_{\max}$.
Let $I_{\ell}=((1+\delta)^{{\ell}-1},(1+\delta)^{\ell}]$, for $1 \leq \ell\leq L$, and $\mathcal{L} = \{1,2,\ldots, L\}$. Note that, interval $[1,1]$ implies that no job finishes its execution before time 1; in fact, we can assume, without loss of generality, that all processing times are positive integers. Note also that, the number of intervals is polynomial in the size of the instance and in $1/\delta$.
For each processor $i\in \mathcal{P}_b$, task $\mathcal{T}_{k,j}\in b$ and $\ell \in \mathcal{L}$, we introduce a variable $y_{i,k,j,\ell}$ that indicates if task $\mathcal{T}_{k,j}$ is completed on processor $i$ within the time interval $I_{\ell}$.
Furthermore, for each task $\mathcal{T}_{k,j} \in \mathcal{T}$, we introduce a variable $C_{k,j}$ corresponding to its completion time.
For every job $j\in \mathcal{J}$, we also introduce a dummy task $D_j$ with zero processing time on every processor, which has to be processed after the completion of every other task $\mathcal{T}_{k,j}\in b$. Note that, the corresponding integer program is a $(1+\delta)$-relaxation of the original problem.
%
\begin{alignat}{2}
 & LP(b): \text{minimize} \sum_{j \in \mathcal{J}} w_j C_{D_j} \notag\\
 & \text{subject to}: \notag\\
 & \sum_{i\in\mathcal{P}_b, \ell\in \mathcal{L}} y_{i, k, j, \ell} \geq 1,~~~~~~~~~~~~~~~~~~~~~~~~~~~~~~~~\forall \mathcal{T}_{k,j}\in b\label{lp:p1}\\
 & C_{D_j} \ge C_{k, j},~~~~~~~~~~~~~~~~~~~~~~~~~~~~~~~\forall j\in\mathcal{J}, \mathcal{T}_{k,j}\in b\label{lp:p3}\\
 & \sum_{i\in\mathcal{P}_b}\sum_{\ell\in \mathcal{L}} (1 + \delta)^{\ell-1} y_{i, k, j, \ell} \le C_{k, j}~~~~~~~~~~~~~~~\forall \mathcal{T}_{k,j}\in b\label{lp:p4}\\
 & \sum_{\mathcal{T}_{k, j}\in b} p_{i, k, j} \sum_{t \le \ell} y_{i, k, j, t} \le (1 + \delta)^\ell,~~~~~~~~~\forall i \in \mathcal{P}_b, \ell\in \mathcal{L}\label{lp:p5}\\
 & p_{i, k, j} > (1 + \delta)^{\ell} \Rightarrow y_{i, k, j, \ell} = 0,~\forall i \in \mathcal{P}_b, \mathcal{T}_{k,j}\in b, \ell\in \mathcal{L}\label{lp:p7}\\
 & y_{i, k, j, \ell} \ge 0,~~~~~~~~~~~~~~~~~~~~~~~~~~\forall i \in \mathcal{P}_b, \mathcal{T}_{k,j}\in b, \ell\in \mathcal{L}\label{lp:p8}
\end{alignat}

Our objective is to minimize the sum of weighted completion times of all jobs. Constraints~(\ref{lp:p1}) ensure that each task is completed on a processor of the set $\mathcal{P}_b$ in some time interval. Constraints~(\ref{lp:p3}) assure that for each job $j\in \mathcal{J}$, the completion of each task $\mathcal{T}_{k,j}$ precedes the completion of task $D_j$.
Constraints~(\ref{lp:p4}) impose a lower bound on the completion time
of each task. For each $\ell\in \mathcal{L}$, constraints (\ref{lp:p5}) and (\ref{lp:p7}) are validity constraints which state that the total processing time of jobs that are executed up to an interval $I_{\ell}$ on a processor $i\in\mathcal{P}_b$ is at most $(1+\delta)^{\ell}$, and that if it takes time more than $(1+\delta)^{\ell}$ to process a task $\mathcal{T}_{j,k}$ on a processor $i\in \mathcal{P}_b$, then $\mathcal{T}_{k,j}$ should not be scheduled on $i$, respectively.

Our algorithm, called \algotask, starts from an optimal fractional solution $(\bar{y}_{i,k,j,\ell},\bar{C}_{k,j},\bar{C}_{D_j})$ to $LP(b)$ and, working along the lines of \cite[Section~5]{CorreaSV12}, rounds it to an integral solution corresponding to a feasible $27/2$-approximate schedule of the job set $\mathcal{J}$ on processors $\mathcal{P}_b$.
The idea of \algotask is to partition the set of tasks $\mathcal{T}_{k,j}$ into classes $S(\ell) = \{\mathcal{T}_{k,j}\in b~|~(1+\delta)^{\ell-1} \leq a \bar{C}_{k,j} \leq (1+\delta)^{\ell} \}$, where $\ell \in \{1, \ldots, L\}$ and $a > 1$ is a parameter, according to their (fractional) completion time in the optimal solution of $LP(b)$, and to use \cite[Theorem~2.1]{ST93} for scheduling the tasks in each class $S(\ell)$ independently. In fact, \algotask can be regarded as a generalization of the approximation algorithm in \cite[Section~4]{HallSSW97}, where the objective is to minimize weighted completion time, but each job consists of a single task (see also the discussion in \cite[Section~5]{CorreaSV12}).

More specifically, we first observe that by the definition of $S(\ell)$ and due to constraints (\ref{lp:p1}) and (\ref{lp:p4}), for each task $\mathcal{T}_{k,j} \in S(\ell)$, $\sum_{i\in\mathcal{P}_b}\sum_{t \leq \ell} y_{i, k, j, t} \geq \frac{a-1}{a}$. Otherwise, it would be $\sum_{i\in\mathcal{P}_b}\sum_{t \geq \ell+1} y_{i, k, j, t} > \frac{1}{a}$, which implies $a \bar{C}_{k,j} > (1+\delta)^\ell$. 
Therefore, if we set $y^\ast_{i, j, k, t} = 0$, for all $t \geq \ell+1$, and $y^\ast_{i, j, k, t} = \frac{a}{a-1}\bar{y}_{i,j,k,t}$, for all $t \leq \ell$, we obtain a solution $y^\ast_{i,j,k,t}$ that satisfies the constraints (\ref{lp:p1}), (\ref{lp:p5}), and (\ref{lp:p7}) of $LP(b)$, if the right-hand side of (\ref{lp:p5}) is multiplied by $a/(a-1)$. Therefore, for each $\ell = 1, \ldots, L$, the tasks in $S(\ell)$ alone can be (fractionally) scheduled on processors $\mathcal{P}_b$ with makespan at most $\frac{a}{a-1}(1+\delta)^\ell$. Now, using \cite[Theorem~2.1]{ST93}, we obtain an integral schedule for the tasks in $S(\ell)$ alone with makespan at most $(\frac{a}{a-1}+1)(1+\delta)^\ell$. By the definition of $S(\ell)$, in this integral schedule, each task $\mathcal{T}_{k,j} \in S(\ell)$ has a completion time of at most $a(\frac{a}{a-1}+1)(1+\delta) \bar{C}_{k,j}$. Therefore, if we take the union of these schedules, one after another, in increasing order of $\ell = 1, \ldots, L$, the completion time of each job $j$ is at most $a(\frac{a}{a-1}+1+\frac{1}{\delta})(1+\delta) \bar{C}_{D_j}$. Choosing $a = 3/2$ and $\delta = 1/2$, we obtain that:

\begin{theorem} \cite{CorreaSV12} \label{thm:maporreduce}
 \algotask is a $27/2$-approximation for scheduling a set of Map tasks (resp. Reduce tasks) on a set of unrelated processors $\mathcal{P}_\mathcal{M}$ (resp. $\mathcal{P}_\mathcal{R}$), in order to minimize their total weighted completion time.
\end{theorem}


\subsection{Merging task schedules}
Let $\sigma_{\mathcal{M}}, \sigma_{\mathcal{R}}$ be two schedules computed by two runs of \algotask, for $b=\mathcal{M}$ and $b=\mathcal{R}$, respectively. Let also $C_{j}^{\sigma_{\mathcal{M}}}=$ $\max_{\mathcal{T}_{j,k}\in \mathcal{M}}\{C_{k,j}\},C_{j}^{\sigma_{\mathcal{R}}} =$ $\max_{\mathcal{T}_{j,k}\in \mathcal{R}}\{C_{k,j}\}$ be the completion times of all the Map and all the Reduce tasks of a job $j\in \mathcal{J}$ within these schedules, respectively. Depending on these completion time values, we assign each job $j\in \mathcal{J}$ a \emph{width} equal to $\omega_j = \max\{C_{j}^{\sigma_{\mathcal{M}}}, C_{j}^{\sigma_{\mathcal{R}}}\}$. 

\begin{algorithm}[h!]
\mrs
\begin{algorithmic}[1]
\State Assign the tasks in $\mathcal{M}\cup\mathcal{R}$ on the same processors as in schedules $\sigma_{\mathcal{M}}$ and $\sigma_{\mathcal{R}}$ respectively.
\For {each job $j\in \mathcal{J}$}
\State Fix $\omega_j = \max\{C_{j}^{\sigma_{\mathcal{M}}}, C_{j}^{\sigma_{\mathcal{R}}}\}$ to be the width job $j$
\EndFor
\For {each time $t$ where a processor $i\in \mathcal{P}$ becomes available}
\If {$i=\mathcal{P}_\mathcal{M}$}
\State Among the unscheduled Map tasks in $i$, schedule task $\mathcal{T}_{k,j}\in \mathcal{M}$ with the smallest $\omega_j$, with processing time $p_{i,k,j}$.
\Else
\State Among the unscheduled Reduce tasks, which have $\omega_j > t$, schedule  task $\mathcal{T}_{k,j}\in \mathcal{R}$ with the smallest $\omega_j$, with processing time $p_{i,k,j}$.
\EndIf
\State Let $C_{k,j}$ be the completion time of task $\mathcal{T}_{k,j}$.
\EndFor
\For {each job $j\in \mathcal{J}$}
\State Compute the completion time $C_j = \max_{\mathcal{T}_{k,j}\in{\mathcal{R}}}C_{k,j}$.
\EndFor
\end{algorithmic}
\end{algorithm}

\mrs computes a feasible schedule by processing, in each time instant where a processor $i\in\mathcal{P}_b$ becomes available, either the Map task, assigned to $i\in \mathcal{P}_{\mathcal{M}}$ in $\sigma_\mathcal{M}$, with the minimum width, or the available (w.r.t. its release time $\omega_j$) Reduce task, assigned to $i\in\mathcal{P}_{\mathcal{R}}$ in $\sigma_\mathcal{R}$, with the minimum width.

Extending the analysis in~\cite{MoseleyDKS11}, we are able to prove that:

\begin{theorem} \label{thm:mrs}
\mrs is a $54$-approximation for the \mr problem.
\end{theorem}
\begin{proof}
  First, we have to prove that the schedule computed by the \mrs algorithm is a non-preemptive one. This is obvious for the Map tasks, while in case of Reduce tasks the only way to have preemption is to have a task $\mathcal{T}_{r_1,j}$ that is not scheduled by the time a task $\mathcal{T}_{r_2,j}$ with higher width is executed. But this cannot happen because if $\mathcal{T}_{r_2,j}$ has higher width, then it will be available after $\mathcal{T}_{r_1,j}$ and our algorithm will schedule first $\mathcal{T}_{r_1,j}$ thus, a contradiction. Therefore, by execution of \mrs it is clear that all tasks are executed non-preemptively, while all Map tasks are scheduled only on the Map processors $\mathcal{P}_\mathcal{M}$ and all Reduce tasks only on the Reduce processors $\mathcal{P}_\mathcal{R}$.

  Now, we have to prove that the resulting schedule respects the precedence between Map and Reduce tasks. Therefore, we have to prove that a Map task with width $\omega_j$ finishes before time $\omega_j$. This means that the corresponding Reduce tasks will be executed afterwards since their release time is $\omega_j$. For the sake of contradiction we assume that there is a map task $\mathcal{T}_{m_1,j}$ with width $\omega_j$ finishing by time $t > \omega_j$. It is obvious that the schedule has no idle time and therefore in the time interval $[0, t]$ the processor $i$ of task $\mathcal{T}_{m_1,j}$ processes tasks with width at most $\omega_j$. However, by definition of width this means that in schedule $\sigma_{\mathcal{M}}$ the processor $i$ processes more than $\omega_j$ volume of work in less than $\omega_j$ time which gives us a contradiction.

Using the same argument as in the Map case, we can prove that the completion time of each Reduce task is upper bounded from $r + \omega_j$, where $r$ is the release time of the task in $\sigma$. Moreover, as we note, $r \leq \omega_j$ and thus $C_j^{\sigma} \le 2\omega_j= 2\max\{C_{j}^{\sigma_{\mathcal{M}}},C_{j}^{\sigma_{\mathcal{R}}}\}$.
  Now, let $C^{OPT}_j$ be the completion time of job $j$ in the overall optimal schedule and let $C^{OPT_{\mathcal{M}}}_j, C^{OPT_{\mathcal{R}}}_j$ be its completion time in the optimal separate schedules of the Map and the Reduce tasks.
  Applying Theorem \ref{thm:maporreduce} and using the fact that $\sum_j C^{OPT}_j \ge \sum_j C^{OPT_{\mathcal{M}}}_j$ and $\sum_j C^{OPT}_j \ge \sum_j C^{OPT_{\mathcal{R}}}_j$,
  the theorem follows.\qed
\end{proof}

\noindent{\em Remark.} If the two sets of processors, $\mathcal{P}_\mathcal{M}, \mathcal{P}_\mathcal{R}$, are not necessarily disjoint (or even if they coincide with each other), then by setting $\omega_j= C_{j}^{\sigma_{\mathcal{M}}} + C_{j}^{\sigma_{\mathcal{R}}}$ and applying a similar analysis, we can yield the same result as in Theorem~\ref{thm:mrs}.

\section{\msr problem}

In the \mr problem of the previous section the Reduce phase of each job can start executed once its Map phase is  finished. However, in real systems there is a significant cost for the key-value pairs with the same key  to be transmitted  to the corresponding single reduce task. In this section, inspired by~\cite{ChenKL12},  we incorporate the data shuffle phase in our model. To this end, we introduce a number of \emph{Shuffle tasks} for each Map task that simulate this transmission of the key-value pairs from a Map to the corresponding Reduce tasks. In contrast to~\cite{ChenKL12}, where the assignment of Shuffle tasks to processors is fixed, we consider a flexible model and study two different variants. In the first variant, each Shuffle task is executed on the same processor with its corresponding Reduce task, while in the second one, we consider a different set of processors executing the Shuffle tasks. For both variants, we present $O(1)$-approximation algorithms.

Note that the number of different keys is in general greater than the number of the Reduce processors available, and in this case a Reduce task receives all key-value pairs of some different keys. Although not all Reduce tasks receive key-value pairs from each Map task, we may assume without loss of generality that this is the case by simply setting  the transmission time of the corresponding Shuffle tasks  equal to zero. We also assume that only a single key-value pair can be transferred to a Reduce processor at any time  and moreover, the transmission process cannot be interrupted. Thus, since the key-value pairs allocated to the same Reduce task cannot be transmitted in parallel, we can assume that all key-value pairs from a Map task that have been assigned to the same Reduce task can be considered as a single Shuffle task. Hence, the number of Shuffle tasks per Map task equals the number of the Reduce tasks.


The following properties summarize the above discussion for the \msr problem:\\[1ex]
{\em
\noindent{\bf Properties}\\[1ex]
(i) Each Shuffle task cannot start its execution before the completion of its corresponding Map task.\\
(ii) For every Map task of a job, there are as many Shuffle tasks as the job's Reduce tasks. Some of them may have zero processing time, indicating that no key-value pairs are transmitted from the corresponding Map task to the corresponding Reduce task).\\
(iii) Each Shuffle task is executed non-preemptively.\\
(iv) Shuffle tasks that are transmitting to the same Reduce processor must not overlap with each other.\\
}

To present our algorithms for the \msr problem we introduce some additional notation.
For each Map task $\mathcal{T}_{k,j}\in \mathcal{M}$ of a job $j\in \mathcal{J}$, we introduce a set of Shuffle tasks $\mathcal{T}_{r,k,j}$, $1 \leq r\leq \tau_j = |\{\mathcal{T}_{k,j}\in \mathcal{R}\}|$, where $\tau_j$ is the number of Reduce tasks of job $j$. We denote by $\mathcal{H}$ the set of Shuffle tasks; note that for each Map task of a job there is a bijection between its Shuffle tasks and the job's Reduce tasks. Each Shuffle task $\mathcal{T}_{r,k,j}\in \mathcal{H}$ is associated with a transfer time $t_{r,k,j}$, which is independent of the processor assignment. In Fig.\ref{fig:shuffle}(i) we depict a MapReduce job $j$, as formed after the introduction of the Shuffle tasks.

\begin{figure}[h!]
\begin{center}
\includegraphics[scale=0.56]{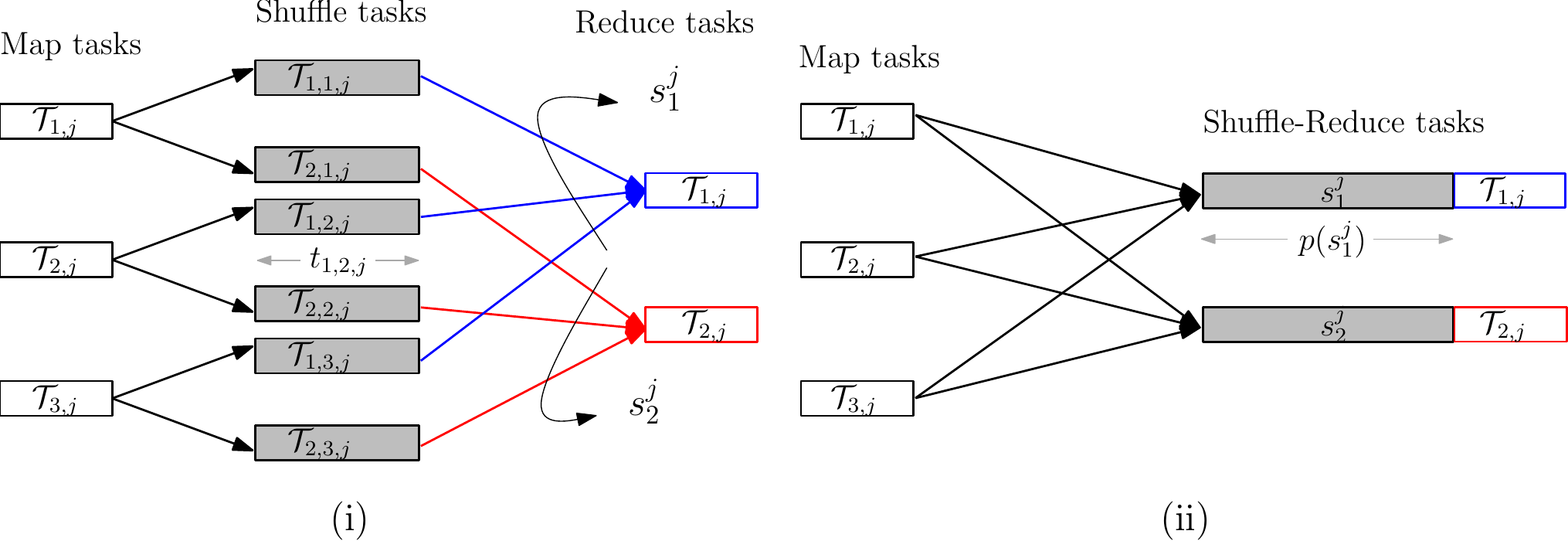}
\end{center}
\caption{(i) Shuffle tasks and their precedence constraints with the Map tasks and Reduce tasks of a job $j$ that comprises three Map tasks and two Reduce tasks and (ii) Precedence constraints among Map tasks and Shuffle-Reduce tasks.}
\label{fig:shuffle}
\end{figure}

\subsection{The Shuffle Tasks are Executed on their Reduce Processors}
When the Shuffle tasks are executed on the same processors with its corresponding Reduce tasks,
our algorithm proceeds into steps as for the \mr problem: a) It computes a $27/2$-approximate schedule for the Map Tasks and a $27/2$-approximate schedule for the
Shuffle-Reduce tasks, with respect to the task Properties (iii)-(iv) and b) it merges the two schedules in a $54$-approximate schedule for the Map-Shuffle-Reduce problem, with respect to the precedence between Map, Shuffle and Reduce tasks.

The key element of our algorithm is the integration of the Shuffle phase into the Reduce phase. In this direction, we consider a Reduce task $\mathcal{T}_{r,j}$ of a job $j$ and let $s_j^r = \{\mathcal{T}_{r,k,j}~|~\mathcal{T}_{k,j}\in \mathcal{M}\}$ be the set of Shuffle tasks that must complete before task $\mathcal{T}_{r,j}$ starts its execution. As the tasks in $s_j^r$ will be executed in the same processor as Reduce task $\mathcal{T}_{r,j}$. Then, we are able to prove the following

\begin{lemma}
There is an optimal schedule of Shuffle tasks and Reduce tasks on processors of the set $\mathcal{P}_{\mathcal{R}}$ such that:\\
(i) There are no idle periods and\\
(ii) All Shuffle tasks in $s_j^r$ are executed together and complete exactly before the Reduce task $\mathcal{T}_{r,j}$ starts its execution.\label{le:s-ropt}
\end{lemma}

\begin{proof}
(i) Consider a feasible schedule $\sigma$, then there are three cases in which an idle time can occur: either between the execution of two Shuffle tasks or two Reduce tasks or between a Shuffle and a Reduce task. Since all Shuffle tasks and Reduce tasks are assumed to be available from time zero and there are no precedence constraints among only Shuffle tasks or only Reduce tasks, skipping the idle times in the first two cases only decreases the objective value of $\sigma$. For the third case, it suffices to notice that since Shuffle tasks precede their corresponding Reduce tasks, by skipping the idles we decrease the completion time of the Reduce tasks and thus the objective value of $\sigma$. Hence, $\sigma$ can be transformed into a schedule of less or equal total weighted completion time.\\
(ii) Again we consider a schedule $\sigma$ that violates the claim and has the last Reduce task $\mathcal{T}_{k,j}$ of a job $j$ completed on some processor $i \in \mathcal{P}_{\mathcal{R}}$. If we fix the completion time of $\mathcal{T}_{k,j}$ and shift all Shuffle tasks in $s_j^r$ to execute just before $\mathcal{T}_{k,j}$, consecutively and in arbitrary order, then, the completion time of $j$ remains unchanged, while that of every task preceding $\mathcal{T}_{k,j}$ in $\sigma$ may decrease. Thus, after a finite number of shifts, $\sigma$ can be transformed into a schedule of less or equal objective value. \qed\end{proof}

By Lemma \ref{le:s-ropt} we are able to reformulate our input so as to incorporate the execution of Shuffle tasks of each job into the execution of its Reduce tasks. More specifically, for each Reduce task $\mathcal{T}_{r,j}$ of a job $j$, for $1\leq r\leq \tau_j$, we increase its processing time $p_{i,r,j}$, on each processor $i\in\mathcal{P}_\mathcal{R}$, by a quantity equal to the total processing time of the Shuffle tasks in $s_j^r$, i.e., $p(s_j^r) = \sum_{\mathcal{T}_{r,k,j}\in s_j^r}t_{r,k,j}$. Let $p_{i,r,j}' = p_{i,r,j}  + p(s_j^r)$ be the increased processing time for each task $\mathcal{T}_{r,j}\in \mathcal{R}$ on processor $i\in\mathcal{P}_\mathcal{R}$, referred as Shuffle-Reduce task.
Let $\mathcal{R}_{\mathcal{H}}$ be the new set of Shuffle-Reduce tasks.
Then, by running Algorithm \textsc{TaskScheduling}($\mathcal{R}_{\mathcal{H}}$) and applying Theorem \ref{thm:maporreduce} we compute a $27/2$-approximate schedule for scheduling the Shuffle-Reduce tasks of $\mathcal{R}_{\mathcal{H}}$. It is not difficult to prove that a schedule produced by \textsc{TaskScheduling}($\mathcal{R}_{\mathcal{H}}$), satisfies Properties (iii)-(v) and thus it is feasible for scheduling Shuffle-Reduce tasks.

In order to merge the two obtained schedules (the one for the Map tasks with the one for Shuffle-Reduce tasks) we note that it suffices to consider the same  precedence constraints, for Map tasks and Shuffle-Reduce tasks, as the ones among Map tasks and Reduce tasks (see Fig.\ref{fig:shuffle}(ii)). The latter dependencies are clearly more general than the precedence constraints between Map tasks and Shuffle tasks of each job (each Shuffle task $\mathcal{T}_{r,k,j}$ cannot start executing before the completion of Map task $\mathcal{T}_{k,j}$) since in order to start the execution of all Shuffle tasks in $s_j^r$ we have to wait for all Map tasks $\mathcal{T}_{k,j}$ of job $j$ to complete. However, it satisfies Property (i), and as we note $C^{OPT(b)}_j$ is a lower bound on $C^{OPT}_j$\footnote{Where $C^{OPT}_j$ is the completion time of job $j$ in the overall optimal schedule and $C^{OPT(b)}_j$ the completion time in optimal schedules of either the Map tasks or the Shuffle-Reduce tasks separately.} for any kind of precedence constraints between Map tasks and Shuffle-Reduce tasks and thus by running \mrs we yield that:
\begin{theorem}
\mrs is a $54$-approximation for the \msr problem.
\end{theorem}

\subsection{The Shuffle Tasks may be Executed on Different Reduce Processors}

When the Shuffle tasks are executed on different processors,  we prove that we lose only a factor of $2$ in the approximation ratio of the Shuffle-Reduce schedule. We assume that for any Reduce processor $i\in\mathcal{P}_{\mathcal{R}}$, there exits an \textit{input} processor which receives data from the Map processors. Therefore, the input processor executes the Shuffle tasks that correspond to the Reduce tasks which have been assigned to $i$. We call the set of input processors $\mathcal{P}_{\mathcal{S}}$.  Then, we can prove the following.

\begin{lemma}
  Consider two optimal schedules $\sigma$ and $\sigma'$  of Shuffle tasks and Reduce tasks on processors of the set $\mathcal{P}_{\mathcal{R}} \cup \mathcal{P}_{\mathcal{S}}$ and on processors of the set $\mathcal{P}_{\mathcal{R}}$ respectively. Let also $C^{\sigma}_{k,j},  C^{\sigma'}_{k,j}$ be the completion times of any Reduce task $\mathcal{T}_{k,j}$ in $\sigma$ and $\sigma'$ repsectively.
Then, it holds that $C^{\sigma'}_{k,j}\leq 2C^{\sigma}_{k,j}$.\label{le:shuffle_dproc}
\end{lemma}

\begin{proof}
    We start with optimal schedule $\sigma$ on the $\mathcal{P}_{\mathcal{R}} \cup \mathcal{P}_{\mathcal{S}}$ processors. We fix a Reduce processor $i^r$, the corresponding input  processor $i^s$ and a Reduce task $\mathcal{T}_{k,j} \in \mathcal{R}$ of a job $j \in \mathcal{J}$. We build the schedule $\sigma'$ on the $i^r$ processor by executing the Reduce tasks in the same order as in $\sigma$ and just before a Reduce task, we execute the corresponding Shuffle tasks. Let $B(k)$ be the set of Reduce tasks executed on processor $i^r$, before $\mathcal{T}_{k,j}$ and $Sh(k)$ the set of the shuffle tasks that correspond to the Reduce tasks $B(k) \cup \{\mathcal{T}_{k,j}\}$ . Then, we have that
    \[C^{\sigma'}_{k,j} = \sum_{\mathcal{T}_{l,j} \in B(k)} p_{i^r, l,j} + \sum_{\substack{\mathcal{T}_{q,l,j} \in Sh(k)\\1\le q\le \tau_j}} t_{q, l,j},\]
    which holds since in $\sigma'$ there is no idle time, as already shown in Lemma \ref{le:s-ropt}. Moreover, since both $B(k)$ and $Sh(k)$ have to complete before $\mathcal{T}_{k,j}$ in $\sigma$, we have that
    \[C^{\sigma}_{k, j} \ge \max \left\{ \sum_{\mathcal{T}_{l,j} \in B(k)} p_{i^r, l,j}, \sum_{\substack{\mathcal{T}_{q,l,j} \in Sh(k)\\1\le q\le \tau_j}} t_{q, l,j} \right\}\]
    and therefore $C^{\sigma'}_{k, j} \le 2 C^{\sigma}_{k,j}$. \qed
\end{proof}

Therefore, if we assume the existence of the $\mathcal{P}_{\mathcal{S}}$ processors then, combining Lemma \ref{le:shuffle_dproc} with Theorem \ref{thm:maporreduce}  we yield a $27$-approximation algorithm for scheduling the Shuffle-Reduce tasks.

Then, by running \mrs  in order to combine this schedule with the schedule of the Map tasks, using the same analysis as before, we get the next corollary. Note that the Shuffle tasks here form a special third stage in the FFS problem.

\begin{corollary}
\mrs is a $81$-approximation for the \msr problem, when the Shuffle tasks run on different processors of the Reduce tasks.
\end{corollary}

\section{Conclusions}

We presented constant-approximation algorithms for scheduling a set of MapReduce jobs on unrelated processors in order to minimize their total weighted completion time.
These are the first constant-approximation algorithms for a general setting of the FFS problem while also, according to our knowledge, this is the most general theoretical model for MapReduce scheduling that have been studied so far.
%
%

An interesting direction for future work concerns the online case of the problem.
As noticed in~\cite{MoseleyDKS11}, even when preemption is allowed, resource augmentation is essential for a reasonable competitive ratio. However, the idea of task preemption in MapReduce implementations is usually quite different from that in classical CPU scheduling. More specifically, when a task is suspended, it does not resume at a later time, but it is forced to start over again (see e.g., \cite{ZahariaBSESS09}). 
This fact, reflects on different online scheduling models, e.g., the preemption-restart~\cite{ShmoysWW95}. 
\end{document}